\documentclass[10pt,a4paper]{article}
\usepackage{mathrsfs,amsmath}
\usepackage{amsfonts}
\usepackage{ctable} 
\usepackage{amssymb}
\usepackage{multicol}
\usepackage{enumerate}
\usepackage[T1]{fontenc}
\usepackage{lmodern}
\usepackage{graphicx}
\usepackage{epsfig}
\usepackage{graphicx}
\usepackage{array}
\usepackage{color}
\usepackage{hyperref}
\hypersetup{colorlinks=true, linkcolor=black}
\usepackage{enumitem}
\usepackage{pgf,tikz}
\usetikzlibrary{arrows}
\usepackage{lastpage}
\def\IC{{\mathbb C}}
\def\ra{\rangle}
\def\la{\langle}

\makeatletter
\newcommand*\dashline{\rotatebox[origin=c]{90}{$\dabar@\dabar@\dabar@$}}
\makeatother
\newlength{\Oldarrayrulewidth}

\usepackage{amsthm}
\theoremstyle{plain}
\newtheorem{proposition}{Proposition}[section]
\newtheorem{theorem}[proposition]{Theorem}
\newtheorem{lemma}[proposition]{Lemma}

\theoremstyle{definition}

\newtheorem{remarks}[proposition]{Remarks}
\theoremstyle{remark}

\makeatletter
\def\Ddots{\mathinner{\mkern1mu\raise\p@
\vbox{\kern7\p@\hbox{.}}\mkern2mu
\raise4\p@\hbox{.}\mkern2mu\raise7\p@\hbox{.}\mkern1mu}}
\makeatother      
\usepackage{tikz}

\def\cG{{\cal G}}

\title{Decomposition of quantum gates}
\author{Chi Kwong Li and Diane Christine Pelejo\\
Department of Mathematics, 
College of William and Mary, 
Williamsburg, VA 23187, USA \\
E-mail: ckli@math.wm.edu, dcpelejo@gmail.com}
\date{}
\begin{document}
\baselineskip 14.8pt
\maketitle

\begin{abstract}
A recurrence scheme is presented to decompose an $n$-qubit  
unitary gate to the product of no more than
$N(N-1)/2$ single qubit gates with small 
number of controls, where $N = 2^n$. 
Detailed description of the recurrence steps and formulas for the number
of $k$-controlled single qubit gates in the decomposition are given.
Comparison of the result to a previous scheme is presented, and
future research directions are discussed. 
\end{abstract}

\section{Introduction} 

The foundation of quantum computation \cite{nc} involves
the encoding of computational tasks into the temporal
evolution of a quantum system. A register of
qubits, identical two-state quantum systems, is employed, and
quantum algorithms can be described by unitary
transformations and projective measurements acting on
the state vector of the register. In this
context, unitary matrices (transformations)
are called quantum gates.
Mathematically, a two-state quantum system has vector states $|\psi\ra$
in $\IC^2$, known as qubits.  The two vectors in 
the standard basis $\{|0\ra, |1\ra\}$ 
for $\IC^2$ correspond to two physically measurable quantum states.
An $n$-qubit system containing registers of $n$-qubits 
has vector states in the Euclidean space 
$\mathbb{C}^2\otimes \cdots \otimes \mathbb{C}^2=(\mathbb{C}^2)^{\otimes n}$
with basis vectors
$$|i_n \cdots i_1\ra = |i_n\ra \otimes \cdots \otimes |i_1\ra, 
\qquad i_1, \dots, i_n \in \{0,1\}$$
corresponding to the $2^n$ physically measurable states.
 
For a single qubit, one can use quantum gates corresponding to unitary
transformations to manipulate (transform) the qubit. For an $n$-qubit system
with large $n$,  it is challenging and expensive to implement quantum gates.
One often has to decompose a general quantum gate  into the product
of simple/elementary  unitary gates which can be readily created physically. 
For a discussion on decomposing a unitary matrix into sets of elementary quantum gates, 
see, for example,  \cite{dask}, \cite{DiV},  \cite{mdel}, \cite{slepoy}, and their references.
By elementary linear algebra, it is known that every $N\times N$ 
unitary matrix can be written as the product of no more than 
$N(N-1)/2$ 2-level unitary matrices (Given's transforms), i.e., unitary
matrices obtained from the identity matrix by changing a $2\times 2$ principal submatrix.
For example, if $U \in M_4$ is unitary, then there are unitary matrices of the form
$$U_1 = {\scriptsize \begin{pmatrix} 
1 & 0 & 0 & 0 \cr
0 & 1 & 0 & 0 \cr
0 & 0 & * & * \cr
0 & 0 & * & * \cr\end{pmatrix}} \quad
U_2 = {\scriptsize\begin{pmatrix} 
1 & 0 & 0 & 0 \cr
0 & * & * & 0 \cr
0 & * & * & 0 \cr
0 & 0 & 0 & 1 \cr\end{pmatrix}} \quad 
U_3 = {\scriptsize \begin{pmatrix} 
* & * & 0 & 0 \cr
* & * & 0 & 0 \cr
0 & 0 & 1 & 0 \cr
0 & 0 & 0 & 1 \cr\end{pmatrix}}$$
so that $U_1U$ has a zero $(4,1)$ entry, $U_2U_1U$ has zero  entries at the $(4,1)$ and 
$(3,1)$ positions, and $U_3U_2U_1U$ has zero entries at the $(4,1), (3,1), (2,1)$ 
positions, and $(1,1)$ entry equal to one. Because $U_3U_2U_1U$ is unitary, 
it will be of the form $[1] \oplus \tilde U$ with $\tilde U \in M_3$.
We can then find unitary matrices of the form
$$
U_4 = {\scriptsize \begin{pmatrix} 
1 & 0 & 0 & 0 \cr
0 & 1 & 0 & 0 \cr
0 & 0 & * & * \cr
0 & 0 & * & * \cr\end{pmatrix}} \quad
U_5 = {\scriptsize \begin{pmatrix} 
1 & 0 & 0 & 0 \cr
0 & * & * & 0 \cr
0 & * & * & 0 \cr
0 & 0 & 0 & 1 \cr\end{pmatrix}} \quad 
U_6 = {\scriptsize \begin{pmatrix} 
1 & 0 & 0 & 0 \cr
0 & 1 & 0 & 0 \cr
0 & 0 & * & * \cr
0 & 0 & * & * \cr\end{pmatrix}}$$
so that $U_5U_4U_3U_2U_1U$ has the form $I_2 \oplus V$ with $V \in M_2$
and $U_6 \dots, U_1 U = I_4$.
It follows that $U = U_1^* \cdots U_6^*.$ 

In the context of quantum information science, not all 2-level unitary
matrices are easy to implement. In this context, one considers
matrices of sizes $N = 2^n$ labeled by binary sequences
$i_n \cdots i_1 \in \{0,1\}^n$
corresponding to the measurable quantum state
$|i_n \cdots i_1 \ra$. Then certain two level unitary matrices
correspond to quantum operations acting on the $j$th qubit
provided the other 
qubits $|i_n\ra, \dots, |i_{j+1}\ra, |i_{j-1}\ra, \dots, |i_1\ra$
assume specified values in $\{|0\ra, |1\ra\}$.
These are known as the fully controlled qubit gates.
For example, when $n = 2$, we label the rows and columns of matrices
by $00,01,10,11$. There are four types of fully-controlled 2-qubit gates:
$$
(0V): {\scriptsize \begin{pmatrix} 
v_{11} & v_{12} & 0 & 0 \cr
v_{21} & v_{22} & 0 & 0 \cr
0 & 0 & 1 & 0 \cr
0 & 0 & 0 & 1 \cr\end{pmatrix}} \quad
(1V): {\scriptsize \begin{pmatrix} 
1 & 0 & 0 & 0 \cr
0 & 1 & 0 & 0 \cr
0 & 0 & v_{11} & v_{12} \cr
0 & 0 & v_{21} & v_{22} \cr\end{pmatrix}} \quad
(V0): {\scriptsize \begin{pmatrix} 
v_{11} &0 &  v_{12} & 0 \cr
0 & 1 & 0 & 0 \cr
v_{21} & 0 & v_{22} & 0 \cr
0 & 0 & 0 & 1 \cr\end{pmatrix}} \quad 
(V1): {\scriptsize \begin{pmatrix} 
1 & 0 & 0 & 0 \cr
0 & v_{11} & 0 & v_{12} \cr
0 & 0 & 1 & 0 \cr
0 & v_{21} & 0 & v_{22} \cr\end{pmatrix}} 
$$
with the unitary $V = {\small \begin{pmatrix} 
v_{11} & v_{12} \cr
v_{21} & v_{22} \cr \end{pmatrix}} \in M_2$.
In particular,  
a $(0V)$-gate corresponds to the unitary operator 
$$a|00\ra + b|01\ra + c|10\ra + d|11\ra \mapsto 
|0\ra V(a|0\ra + b|1\ra) + |1\ra (c|0\ra + d|1\ra),$$
which will only change the part of the vector state with the first 
qubit equal to $|0\ra$. Similarly, 
a $(1V)$-gate corresponds to the unitary operator 
$$a|00\ra + b|01\ra + c|10\ra + d|11\ra \mapsto 
|0\ra (a|0\ra + b|1\ra) + |1\ra V(c|0\ra + d|1\ra),$$
which will only change the part of the vector state with the first qubit equal to 
$|1\ra$. The $(V0)$-gate and $(V1)$-gate have the same physical interpretation.
One can associate the 4 types of controlled qubit gates with the following 
circuit diagrams:
\begin{center}
$ (0V):$  \hskip 0.5cm
\begin{tikzpicture}[line cap=round,line join=round,x=1.0cm,y=1.0cm]
\draw (0,0) -- (0.3,0);
\draw (0.3,-0.2) -- (0.3,0.2) -- (0.7,0.2) -- (0.7,-0.2) -- cycle;
\draw (0.5,0) node {$V$};
\draw (0.7,0) -- (1,0);
\draw (0,-0.5) -- (1,-0.5);
\draw (0.5,-0.5) circle (3pt);
\draw (0.5,-0.2) -- (0.5,-0.6);
\end{tikzpicture}\hskip 1.5cm
$(1V):$ 
\hskip 0.5cm
\begin{tikzpicture}[line cap=round,line join=round,x=1.0cm,y=1.0cm]
\draw (0,0) -- (0.3,0);
\draw (0.3,-0.2) -- (0.3,0.2) -- (0.7,0.2) -- (0.7,-0.2) -- cycle;
\draw (0.5,0) node {$V$};
\draw (0.7,0) -- (1,0);
\draw (0,-0.5) -- (1,-0.5);
\fill (0.5,-0.5) circle (3pt);
\draw (0.5,-0.2) -- (0.5,-0.6);
\end{tikzpicture}\hskip 1.5cm
$(V0):$ \hskip 0.5cm
\begin{tikzpicture}[line cap=round,line join=round,x=1.0cm,y=1.0cm]
\draw (0,0) -- (0.3,0);
\draw (0.3,-0.2) -- (0.3,0.2) -- (0.7,0.2) -- (0.7,-0.2) -- cycle;
\draw (0.5,0) node {$V$};
\draw (0.7,0) -- (1,0);
\draw (0,0.5) -- (1,0.5);
\draw (0.5,0.5) circle (3pt);
\draw (0.5,0.2) -- (0.5,0.6);
\end{tikzpicture}\hskip 1.5cm
$(V1):$
\hskip 0.5cm
\begin{tikzpicture}[line cap=round,line join=round,x=1.0cm,y=1.0cm]
\draw (0,0) -- (0.3,0);
\draw (0.3,-0.2) -- (0.3,0.2) -- (0.7,0.2) -- (0.7,-0.2) -- cycle;
\draw (0.5,0) node {$V$};
\draw (0.7,0) -- (1,0);
\draw (0,0.5) -- (1,0.5);
\fill (0.5,0.5) circle (3pt);
\draw (0.5,0.2) -- (0.5,0.6);
\end{tikzpicture}\hskip 1.5cm
\end{center}
For $n = 3$, we have fully-controlled qubit gates of the types:
$$(00V), (01V), (10V), (11V), (0V0), (0V1), (1V0), 
(1V1), (V00), (V01), (V10), (V11).$$ 
One easily extends this idea and notation to define fully-controlled gates acting
on $n$-qubits.

In \cite{fin} (see also \cite{reb}),
it was shown that one can decompose a quantum gate into the product of 
2-level matrices corresponding to fully-controlled qubit gates.
While fully-controlled qubit gates are relatively simple, it is still not easy
to implement because the qubit gate $V$ can only act on the target bit
after verifying that the other $(n-1)$-qubits satisfy the controlled bits.
As mentioned in \cite{fin},
in practice it is desirable to replace fully controlled qubit gates by 
qubits gates as few controls as possible.
For example, when $n = 2$, the following types of
unitary gates with no controls
$$(*V):  I_2 \otimes V
= {\scriptsize \begin{pmatrix} 
v_{11} & v_{12} & 0 & 0 \cr
v_{21} & v_{22} & 0 & 0 \cr
0 & 0 & v_{11} & v_{12} \cr
0 & 0 & v_{21} & v_{22} \cr\end{pmatrix}}, \qquad 
(V*): V \otimes I_2 = {\scriptsize \begin{pmatrix} 
v_{11} & 0 & v_{12} & 0 \cr
0 & v_{11} & 0 & v_{12} \cr
v_{21} & 0 & v_{22} & 0 \cr
0 & v_{21} & 0 & v_{22} \cr\end{pmatrix}}$$
are easier to implement. 
Note that a $(0V)$-gate is applied on the left 
of a matrix $A\in M_4$, only rows $00$ and $01$ are affected. 
Similarly, a $(1V)$-gate will only affect the $10$ and $11$ gate 
of $A$. However, a $(*V)$-gate and $(V*)$-gate will affect all rows 
of $A$. 

In general, we can consider a $(c_n c_{n-1}\cdots c_1)$-unitary gate
with $c_n, \dots, c_1 \in \{0,1,*,V\}$, where only one of the terms
is $V$, and the number of terms in $\{0,1\}$ is the total number
of controls. For example, a $(11*0V1)$-unitary gate acting on 
6-qubit states has 4 controls,  and the target qubit is the fifth one. 

In \cite{fin}, a recurrence scheme was proposed to decompose a unitary gate 
as the product of controlled qubit gates with small number of controls.
The purpose of this paper is to present another simple recurrence scheme,
which provide an alternative choice for implementation.
Moreover, the ideas and techniques in the construction may be
helpful for further research in this and related problems.

Our paper is organized as follows. 
In Section 2, 
we will illustrate our scheme for the 2-qubit and
3-qubit case, and discuss how it can be extended. 
In Section 3, we  present the general scheme with detailed description
of the implementation steps and explanation of their validity.
In Section 4, we obtain formulas for the number
of $k$-controlled single qubit gates in the decomposition and  
compare our results to those in scheme in \cite{fin}. 
Concluding remarks and future research directions
are mentioned in Section 5.

\section{Two-qubit and Three-qubit cases}

For an $n$-qubit unitary gate $U \in M_N$ with $N = 2^n$,
we will describe a recurrence scheme for generating
controlled single qubit unitary gates 
$U_1, \dots, U_r$ with $r \le N(N-1)/2$ 
such that $U_r \cdots U_1 U = I_N$.
Consequently, $U = U_1^{\dag} \cdots U_r^{\dag}$.

Our scheme is done as follows. Assume we have the reduction scheme for
the $(n-1)$-qubit case.

\medskip\noindent
{\bf Step 1}
Partition $U \in M_n$ into a $2\times 2$ block matrix with each block lying 
in $M_{N/2}$.

\medskip\noindent
{\bf Step 2} Use the scheme of the $(n-1)$-qubit case to help reduce $U$ to the form 
$I_{N/2} \oplus \tilde U$ with $\tilde U \in M_{N/2}.$

\medskip\noindent
{\bf Step 3} Apply the scheme of the $(n-1)$-qubit case to transform $\tilde U$ to $I_{N/2}$. 

\bigskip
In {\bf Step 2}, we need to eliminate the nonzero off-diagonal entries of 
$U$ for the  first $N/2$ columns.  We will do these elimination column by column.
In each column, we first eliminate the off-diagonal entries 
with row indices smaller than $N/2+1$ using the scheme in the $(n-1)$-qubit case.
Then we again use the elimination schemes of lower dimension cases to eliminate the
entries with row indices larger than $N/2$.

\medskip\noindent
{\large{\bf  The  two-qubit gate.}}

\medskip
In the following tables, we indicate the order of the entries to be eliminated in our
scheme, and also the $(c_2c_1)$-gates used to do the elimination.

\medskip
\noindent
Column 1. 
\[\begin{array}{|c|c|c|c|}
\hline 
\hbox{ entries } & \mbox{ \footnotesize\tt (2,1)} & 
\mbox{ \footnotesize\tt (4,1)} & \mbox{ \footnotesize\tt (3,1)} \\\hline
\hbox{ gates } & \mbox{ \footnotesize\tt (*V)} & \mbox{ \footnotesize\tt (1V)} & 
\mbox{ \footnotesize\tt (V*)}\\\hline
\end{array} \]

\noindent
Column 2. 
\[\hskip 1.1cm \begin{array}{|c|c|c|}
\hline 
\hbox{ entries } & \mbox{ \footnotesize\tt (3,2)} & \mbox{ \footnotesize\tt (4,2)} 
\\\hline
\hbox{ gates } & \mbox{ \footnotesize\tt (1V)} & \mbox{ \footnotesize\tt (V1)}\\\hline
\end{array}\]

\noindent
Column 3. 
\[\begin{array}{|c|c|}
\hline 
\hbox{ entries } & \mbox{ \footnotesize\tt (4,3)} \\\hline
\hbox{ gates } & \mbox{ \footnotesize\tt (1V)} \\\hline
\end{array} \hskip 2.3cm\]

Here we first eliminate the $(2,1)$ entry as in the 1-qubit case.  
In a similar manner, annihilate the $(4,1)$ entry, treating it as the second entry of 
the lower left half of the first column. To keep the $(2,1)$ entry zero, we use a gate 
with a $1-control$ in the leftmost bit. Finally we annihilate the $(3,1)$ entry with 
the help of the $(1,1)$ entry. In this case, we can use a control-free gate to do so. 
At this point, the current form of the matrix is $[1]\oplus U'$, where $U'\in M_{3}$. 

Then we move to the second column. We adapt the procedure of eliminating the 
$(4,1)$ and $(3,1)$ entries to eliminate the $(3,2)$ and $(4,2)$ entries. 
The gates used must not change the zero entries in the first column.
After this, the matrix takes the form $I_2 \oplus U_1$ with $U_1 \in M_2$.
We can deal with the matrix $U_1$ as in the 1-qubit case
using a $(1V)$-gate so that the first two rows will not
be affected.

\medskip\noindent
{\large{\bf The three qubit case.}}

In this case, we have  3 types of unitary gates with no control:
$$(**V), \ (*V*), \ (V**);$$  $12$ types of unitary gates
with 1 control ($0$ or $1$) and 1 target qubit: 
$$(0*V),(1*V), (0V*), (1V*),
(*0V), (*1V), (V0*), (V1*), (*V0), (*V1), (V*0), (V*1);$$ 
and 12 types of 
unitary gates with 2 controls and 1 target qubit: 
$$(00V), (01V), (10V), (11V), (0V0), (0V1), (1V0), (1V1), (V00), (V01), (V10), (V11).$$

We execute the reduction scheme for three qubit gates as follows.

\medskip
\noindent
Column 1. 
$$\begin{array}{|c|c|c|c|c|c|c|c|}
\hline
\hbox{ entries } & \mbox{ \footnotesize\tt{(2,1)}} & \mbox{ \footnotesize\tt{(4,1)}} &  \mbox{ \footnotesize\tt{(3,1)}} & \mbox{ \footnotesize\tt(6,1)} & \mbox{ \footnotesize\tt(8,1)} & \mbox{ \footnotesize\tt(7,1)} & \mbox{ \footnotesize\tt (5,1)}\cr
\hline 
\hbox{ gates } & \mbox{ \footnotesize\tt{(**V)}} & \mbox{ \footnotesize\tt{(*1V)}} & \mbox{ \footnotesize\tt{(*V*)}} &\mbox{ \footnotesize\tt(1*V)} & \mbox{ \footnotesize\tt(*1V)} & \mbox{ \footnotesize\tt(1V*)} & \mbox{ \footnotesize\tt(V**)}\cr
\hline
\end{array}$$

\noindent
Column 2. 
$$\hskip .6in \begin{array}{|c|c|c|c|c|c|c|}
\hline
\hbox{ entries } & \mbox{ \footnotesize\tt{(3,2)}} &  \mbox{ \footnotesize\tt{(4,2)}} & \mbox{ \footnotesize\tt(5,2)} & \mbox{ \footnotesize\tt(7,2)} & \mbox{ \footnotesize\tt(8,2)} & \mbox{ \footnotesize\tt(6,2)}\cr
\hline 
\hbox{ gates } & \mbox{ \footnotesize\tt{(*1V)}} & \mbox{ \footnotesize\tt{(*V1)}} & 
\mbox{ \footnotesize\tt(1*V)} & \mbox{ \footnotesize\tt(*1V)} & \mbox{ \footnotesize\tt(1V*)} & \mbox{ \footnotesize\tt(V*1)}\cr
\hline
\end{array}$$

\noindent
Column 3. 
$$\hskip 1.2in \begin{array}{|c|c|c|c|c|c|}
\hline
\hbox{ entries } & \mbox{ \footnotesize\tt{(4,3)}} &  \mbox{ \footnotesize\tt(8,3)} & \mbox{ \footnotesize\tt(6,3)} & \mbox{ \footnotesize\tt(5,3)} & \mbox{ \footnotesize\tt(7,3)}\cr
\hline 
\hbox{ gates } & \mbox{ \footnotesize\tt{(*1V)}} &\mbox{ \footnotesize\tt(1*V) } & \mbox{ \footnotesize\tt(10V)} & \mbox{ \footnotesize\tt(1V*)} & \mbox{ \footnotesize\tt(V1*)}\cr
\hline
\end{array}$$

\noindent
Column 4. 
$$\hskip 1.8in \begin{array}{|c|c|c|c|c|}
\hline
\hbox{ entries } & \mbox{ \footnotesize\tt(7,4)} &  \mbox{ \footnotesize\tt(5,4)} & \mbox{ \footnotesize\tt(6,4)} & \mbox{ \footnotesize\tt(8,4)}\cr
\hline 
\hbox{ gates } & \mbox{ \footnotesize\tt(1*V)} & 
\mbox{ \footnotesize\tt(10V)} & \mbox{ \footnotesize\tt(1V*)} & \mbox{ 
\footnotesize\tt(V11)}\cr
\hline
\end{array}$$

\noindent
Column 5. 
\[\begin{array}{|c|c|c|c|}
\hline 
\hbox{ entries } & \mbox{ \footnotesize\tt(6,5)} & \mbox{ \footnotesize\tt(8,5)} & 
\mbox{ \footnotesize\tt(7,5)} \\\hline
\hbox{ gates } & \mbox{ \footnotesize\tt(1*V)} & \mbox{ \footnotesize\tt(11V)} & \mbox{ 
\footnotesize\tt(1V*)}\\\hline
\end{array}\hskip 5.5cm \]

\noindent
Column 6. 
\[\begin{array}{|c|c|c|}
\hline 
\hbox{ entries } & \mbox{ \footnotesize\tt(7,6)} & \mbox{ \footnotesize\tt(8,6)} 
\\\hline
\hbox{ gates } & \mbox{ \footnotesize\tt(11V)} & \mbox{ \footnotesize\tt(1V1)}\\\hline
\end{array}\hskip 4.2cm \]

\noindent
Column 7. 
\[\begin{array}{|c|c|}
\hline 
\hbox{ entries } & \mbox{ \footnotesize\tt(8,7)} \\\hline
\hbox{ gates } & \mbox{ \footnotesize\tt(11V)}\\\hline
\end{array} \hskip 3cm\]

\begin{remarks} \label{n3}\ 
Here we give some remarks about the reduction of a 3-qubit unitary
gate to help illustrate our recurrence scheme and how it can be extended.
The comments are numbered according to the major steps 1--3 of our 
scheme described in the beginning of this Section.

\medskip\noindent
{\bf Step 1} We partition the $8 \times 8$ unitary matrix into 2-by-2 block matrix
so that each block is $4\times 4$.

\bigskip\noindent
{\bf Step 2} We consider Column 1, 2, 3, 4, 

\medskip\noindent
{\bf For Column 1}, the elimination of $(2,1), (4,1), (3,1)$ 
entries will be done as in the  $4\times 4$ (2-qubit) case by changing the
2-qubit $(c_2c_1)$-gates to $(*c_2c_1)$-gates in these steps.

We then annihilate the $(6,1),(8,1)$ and $(7,1)$ entries 
the same way we annihilated the $(2,1),(4,1)$ and $(3,1)$ entries by 
treating the lower half as a $4\times 4$ matrix.
However, we have to ensure that the $(1,1)$
entry will not interact with the zero entries at the $(2,1), (3,1), (4,1)$ 
positions in these steps.
So, we adapt the 2-qubit $(c_2c_1)$-gates to $(c_3c_2c_1)$-gates,  
we will use the following rule:

\medskip
\centerline{\it 
let $c_3 = 1$ if $(c_2c_1)$ is $(*V)$ or $(V*)$; otherwise,
let $c_3 = *$.}

\medskip\noindent
So, a $(1*V)$-gate can be used to annihilate the $(6,1)$ entry, a $(*1V)$-gate can 
be used to annihilate the $(8,1)$ entry and a $(1V*)$-gate to annihilate the $(7,1)$ 
entry. 
Finally, we can apply a $(V**)$-gate to eliminate the 
the $(5,1)$ entry using the $(1,1)$ entry.  

Note that the $(c_3c_2c_1)$-gates used in the Column 1 satisfy
$c_3, c_2, c_1 \in \{*,1,V\}$ with $c_1 \ne 1$. 
This property will hold for the general case.

Once all off-diagonal
entries in Column 1 are annihilated, we obtain a matrix of the form 
$[1]\oplus U'$, where $U'\in M_{7}$. We can proceed to Column 2.

\bigskip\noindent
{\bf For Column 2}, 
we can annihilate the $(3,2)$ and $(4,2)$ entries using the scheme for
annihilating the second column in the $4\times 4$ case
by changing the
2-qubit $(c_2c_1)$-gates to $(*c_2c_1)$-gates in these steps.
 
Next, we  adapt the scheme of annihilating  the $(6,1),(8,1),(7,1),(5,1)$ 
entries to annihilate the lower half entries of the second column.
Note that it is imperative that the $(6,2)$ entry be the last entry to be 
annihilated since it is the only entry in the lower half of the column that 
can be annihilated using the $(2,2)$ entry. In view of this,
we will change the order of annihilation of the entries to:
$$(5,2), (7,2), (8,2), (6,2).$$ 
If we identify $(1,2,\dots,8)$ with the binary sequences 
$(000, 001, \dots, 111)$, then 
\begin{center}
$(6,8,7,5)$ corresponds to $(101,111,110,100)$, and
$(5,7,8,6)$ corresponds to $(100,110,111,101)$.
\end{center}
\noindent
The conversion can be easily realized by
\begin{eqnarray*} (100,110,111,101) &=& (101,111,110,100) \oplus (001,001,001,001)\cr
&=& (101\oplus 001,111 \oplus 001,110 \oplus 001,100 \oplus 001),
\end{eqnarray*}
where $i_3i_2i_1\oplus j_3j_2j_1$ is an entry-wise addition such that
$0\oplus 0 = 1\oplus 1 = 0$ and $0\oplus 1 = 1 \oplus 0 = 1$.
Note that we will use a similar conversion for columns 3 and 4.

We also need to modify the $(c_3c_2c_1)$-gates used for the annihilation 
of the $(6,1), (8,1), (7,1)$ entries
to annihilate the $(5,2), (7,2), (8,2)$ entries. 
To accomodate the change in the order of annihilation, one must modify any control 
found in $c_1$. We also have to prevent the $(1,1)$ entries interacting 
with the $(2,1), (3,1), (4,1)$ entries, and also prevent the $(2,2)$ entries
interacting with the $(3,2)$ and $(4,2)$ entries. This can be done by making sure that at least one of $c_2$ and $c_3$ is equal to $1$. Thus, we modify $(c_3c_2c_1)$ by the following rules: 

\medskip\centerline{\it change $c_3$ to 1 if none of $c_2, c_3$ is 1; 
change $c_1$ to 0 if $c_1 = 1$.}

\medskip\noindent
However, one sees that applying these rules will not
change the $(c_3c_2c_1)$-gates in view of the fact
that $c_1 \ne 1$. Hence we can use exactly the same set of 
$(c_3c_2c_1)$-gates to eliminate the $(5,2), (7,2), (8,2)$ entries of
Column 2.\footnote{As we will see, 
the same phenomenon will hold for columns 3 and 4, and also for the general case.}
Thus, we will use $(1*V), (*1V), (1V*)$ gates to 
annihilate the $(5,2),(7,2)$ and $(8,2)$ entries, respectively.

To annihilate the $(6,2)$ entry, we need to utilize the nonzero $(2,2)$ entry. 
These two entries correspond to rows $101$ and $001$. This means that the target 
bit of the gate  we need is the third bit (leftmost). Because we do not want to 
change the form of  the upper 
half of the first column, we need to make sure that the the gate is not satisfied by 
$000$ but 
is  satisfied by $001$ and $101$. Thus, we use a $(V*1)$-gate. Once this is done, the 
matrix is 
now reduced to the form $I_2\oplus V''$ where $V''\in M_6$.

\bigskip\noindent
{\bf For Column 3}, 
the $(4,3)$ entry is annihilated using the scheme for the third column 
of the $4\times 4$ case. 

Similar to the case in Column 2, we can adapt the scheme of eliminating
the $(6,1), (8,1), (7,1), (5,1)$ entries to 
annihilate the $(8,3), (6,3), (5,3), (7,3)$ entries.
The conversion $(6,8,7,5)$ to $(8,6,5,7)$ is done by performing
$$(111,101,100,110) = (101,111,110,100) \oplus (010,010,010,010)$$
using the binary number correspondence of the indices.

We also need to modify the $(c_3c_2c_1)$-gates used for the annihilation 
of the $(6,1), (8,1), (7,1)$ entries
to annihilate the $(8,3), (6,3), (7,3)$ entries.
In these steps, we have to prevent the $(1,1)$ entries interacting 
with the $(2,1), (3,1), (4,1)$ entries,  the $(2,2)$ entries
interacting with the $(3,2),(4,2)$ entries, and the 
$(3,3)$ entry interacting with the $(4,3)$ entry.
One can do this by adjusting
the $c_3$ and $c_2$ values in the $(c_3c_2c_1)$-gates used 
for the annihilation of the $(6,1), (8,1), (7,1), (5,1)$ entries
by the following rules: 

\medskip\centerline{\it change $c_3$ to 1 if $c_3$ is not 1; 
change $c_2$ to $0$  if $c_2 = 1$.}

\medskip\noindent
Since $c_3$ is 1, for $i = 1,2,3,4$, the $(i,i)$ entry
will not interact with other $(k,i)$ entries for $1 \le k \le 4$ and $k \ne i$.
Note that a $(c_3c_2c_1)$-gate corresponds to a unitary matrix $\tilde V \in M_8$.
Changing a control bit in the position of $c_2$ corresponds to changing 
$\tilde V$ by a permutation similarity $P^t \tilde V P$, where $P$
corresponds to the change of the basis 
$\{|000\ra, \dots, |111\ra\}$ to 
$\{|010\ra, \dots, |101\ra\}$, here we change $|j_2j_2j_1\ra$ to 
$|j_3 (j_2\oplus 1) j_1\ra$.  Thus, the modified $(c_3c_2c_1)$-gates can be used for
Column 3. We will give a general description of this procedure in the next section.
Here, we obtain the $(1*V), (10V), (1V*)$ gates, which can be used 
to annihilate the $(8,3),(6,3),(5,3)$ entries. 

Finally, to annihilate the $(7,3)$ entry, we use the $(3,3)$ entry. 
Hence, the target bit of the gate we need is the leftmost bit. To avoid changing the 
form of the  first and second columns, we need to use controls that are not 
satisfied by $000$ and $001$ but 
is satisfied by $010$ and $110$. Thus, we use the gate $(V1*)$.

\medskip\noindent
{\bf For Column 4}, 
we need not do anything about the first four entries at this point.

We will adapt the scheme for
the $(6,1), (8,1), (7,1), (5,1)$ entries to 
annihilate the $(7,4),(5,4),(6,4),(8,4)$ entries.
The conversion $(6,8,7,5)$ to $(7,5,6,8)$ is done by performing
$$(110,100,101,111) = (101,111,110,100) \oplus (011,011,011,011)$$
using the binary number correspondence of the numbers.

We adjust the $(c_3c_2c_1)$-gates used 
for  the $(6,1), (8,1), (7,1)$ entries
to annihilate the $(7,4),(5,4),(6,4)$ entries as follows,

\medskip\centerline{
\it change $c_3$ to 1  if $c_3$ is not 1; for $i = 1,2$, change 
$c_i$ to $0$ if $c_i = 1$.}

\medskip\noindent
Note that column 4 is associated to the binary sequence $011$.\footnote{
As we will see in the  next section, we always adjust the gates 
according to the the binary sequence associated to the column index.}
We will obtain the $(1*V),(10V),(1V*)$ gates, which can be used  
to annihilate the $(7,4),(5,4),(6,4)$ entries.\footnote{
Note also that the $(c_3c_2c_1)$-gates are the same as those used in 
Column 3 before the final step. We will also explain this in the next section.}
Finally use a $(V11)$-gate to annihilate the $(8,4)$ entry using the $(4,4)$ 
entry while avoiding any change in the form of the first three columns. 

\bigskip\noindent
{\bf Step 3}
Note that after Column 4 is dealt with, the matrix takes the form 
$I_4\oplus V'$ where  $V'\in M_{4}$. We can then use the scheme for 
the $2$-qubit case to transform $V'$ to $I_4$.
However, to avoid changing the form of the first four columns, 
we need to extend the $(c_2c_1)$-gates used in the $4\times 4$ 
case to $(1c_2c_1)$-gates for the remaining steps. This 
explains the tables for columns 5 to 7. 
\end{remarks}

\section{General Scheme}

In this section, we present the general recurrence scheme for
the annihilation of the off-diagonal entries of an $n$-qubit unitary
gate by adapting the reduction scheme of the $(n-1)$-qubit case.
We will carry out {\bf Steps 1 -- 3}
described at the beginning of Section 2,
As illustrated in the 3-qubit case and explained in Remark 2.1,
{\bf Step 2} of the scheme requires some careful attention. 
For each column $\ell = 1, \dots, N/2$ with $N = 2^n$, 
we can always annihilate the off-diagonal 
entries in the upper half of column $\ell$ using the scheme for 
annihilating the first column for an $(n-1)$ qubit unitary gate. 
One only needs to change a $(c_{n-1}\cdots c_2c_1)$-gate to a 
$(*c_{n-1}\cdots c_1)$-gate.

For the lower half of column $\ell$, we have to refine {\bf Step 2} to the
following steps.

\medskip\noindent
{\bf Step 2.1} For column 1, use the reduction scheme for an $(n-1)$-qubit 
to eliminate the off-diagonal entries in the upper half of the column by
changing the $(c_{n-1} \cdots c_1)$-gates used in the $(n-1)$-qubit
gate case to $(*,c_{n-1}, \dots, c_n)$-gates in these steps.

\medskip
Next, we apply the same scheme to eliminate the entries in
the lower half except for the $(N/2+1,1)$ entry, which will be eliminated last. 
This is done by changing the $(c_{n-1} \cdots c_1)$-gates in the
$(n-1)$-qubit case to $(c_n \cdots c_1)$-gates, where 
\begin{equation}\label{cn}
c_n = \begin{cases} 1 & \hbox{ none of } c_{n-1},\dots, c_1 \hbox{ equals } 1,\cr
* & \hbox{ otherwise}. \cr
\end{cases}
\end{equation}
The $(c_n\dots c_1)$-gate constructed in this way will ensure that
the $(1,1)$ entry will not interact with $(2,1),\ldots (N/2,1) $ entries  when we annihilate the 
$(N/2+j,1)$ entry for $j = 2, \dots, N/2$ because 
$1 \in \{c_n, \dots, c_1\}$.
Finally, apply a $(V * \cdots *)$-gate to annihilate the 
$(N/2+1,1)$ entry.

An easy inductive argument will verify that the $(c_n\cdots c_1)$-gates used in
Column 1 satisfy $c_n, \dots, c_1\in \{*,1,V\}$ with $c_1 \ne 1$.

The annihilation steps of Column 1 can be summarized in the following.

\newpage
\baselineskip 15.2pt
\medskip
\centerline{\bf Procedure 2.1}.

\begin{center}
$\begin{array}{|l|}
\hline
\\
\hbox{Suppose in the } (n-1)-\hbox{qubit case, the off-diagonal
entries in the first column are eliminated in the order of }\\ \\
\hskip 1in 
(b_1,1), \dots, (b_{N/2-1},1) \quad \hbox{ by } \quad 
{\bf C_1}-\hbox{gate}, \dots  
{\bf C_{N/2-1}}-\hbox{gate}. \\ \\
\hbox{Eliminate the entries in the upper half of the Column 1
in the order of }
\\ \\
\hskip 1in 
(b_1,1), \dots, (b_{N/2}-1,1) \quad \hbox{ by } 
\quad (*{\bf C_1})-\hbox{gate}, \dots,
(*{\bf C_{N/2-1}})-\hbox{gate}.
\\ \\
\hbox{For } \mathbf{C}=(c_{n-1} \cdots c_1) \hbox{ let }
\cG(\mathbf {C}) = (c_n c_{n-1} \cdots c_1) \hbox{ with }
c_n \hbox{ satisfying (\ref{cn})}. \\ \\
\hbox{Eliminate the entries in the lower half of the column
in the order of }
\\ \\
\hskip 1in
(d_1,1), \dots, (d_{N/2-1},1) 
\quad \hbox{ by } \quad {\cal G}({\bf C_1})-\hbox{gate}, 
\dots   {\cal G}({\bf C_{N/2-1}})-\hbox{gate}, \\ \\
\hbox{where }  d_i = b_i + N/2 \hbox{ for } i = 1, \dots, N/2-1, and
\hbox{ eliminate the } (N/2+1,1) \hbox{ entry by a }
(V*\cdots*)-\hbox{gate}.  \\  \\
\hline 
\end{array}
$
\end{center}

\bigskip\noindent
{\bf Step 2.2} 
For column $\ell$ with $2 \le \ell \le N/2$, 
we can use the same scheme as that of the $(n-1)$-qubit case to 
eliminate the off-diagonal entries in 
the upper half. 
Then we can adapt the scheme for eliminating the 
entries in the lower half of Column 1 to other columns.
To this end, we need to modify 

\begin{itemize}
\item[(a)] the order of the elimination of the
entries in the lower half so that the last entry 
in the lower half  will be eliminated by the $(\ell,\ell)$ entry.

\item[
(b)] the control gates used to do the elimination so that

\ \ 
(b.i) they
will not affect the zero entries obtained in the previous steps; and 

\ \
(b.ii) they will annihilate the entries in the order prescribed in (a). 
\end{itemize}
To achieve (a) and (b), identify
$k \in \{1, \dots, 2^{n}\}$ with the binary sequence
$\tilde k_n \cdots \tilde k_1
\in \{\underbrace{0\cdots 0}_n,  
\dots, \underbrace{1\cdots 1}_n\}$ so that
$$k = \sum_{j=1}^n \tilde k_j 2^{j-1} + 1.$$ 
For (a), if we annihilate the entries in the lower half of Column 1 
in the order of $(d_1, 1), \dots, (d_{N/2}, 1)$, then we will annihilate the entries
in the lower half  of column $\ell$ in the order of
$${(d_1 \oplus \ell, \ell), \dots, (d_{N/2} \oplus \ell, \ell)},$$
where the binary sequence of $d_j \oplus \ell$ is obtained by
entry-wise addition $\oplus$ (without carried digits) 
of the two binary sequences of $d_j$ and $\ell$
such that $0\oplus 0 = 1 \oplus 1 = 0$ and 
$0\oplus 1 = 1 \oplus 0 = 1$.\footnote{
For instance, the binary form of $f_2(d_i)$ is the sum of (using $\oplus$)
the binary sequence  $(0\cdots 01)$ and the binary form of $d_i$; 
the binary form of $f_3(d_i)$ is the sum of 
the binary sequence  $(0\cdots 010)$ and the binary form of $d_i$;
$\dots$, and  
the binary form of $f_{N/2}(d_i)$ is the sum of 
the binary sequence  $(01\cdots 1)$ and the binary form of $d_i$.}
Note that $d_{N/2} = N/2+1$, and hence $d_{N/2} \oplus \ell = N/2+\ell$,
so that $(N/2+\ell,\ell)$ is the last entry in the lower half of Column $\ell$
to be eliminated.

For (b), suppose $2^{m-1} < \ell \le 2^m$ 
with $m \in \{1, \dots, n-1\}$ and 
$\ell = \sum_{j=1}^{n}\tilde \ell_j 2^{j-1} + 1$.
We  adjust the $(c_n\cdots c_1)$-gate used to annihilate
the $(d_i,1)$ entry with $N/2+1 \le d_i < N$ to the 
$(\tilde c_n \cdots \tilde c_1)$-gates for annihilating the 
$(d_i\oplus \ell,\ell)$ entry as follows, where
\begin{equation} \label{tcj}
\tilde c_j = 
\begin{cases} 
1 & \hbox{ if } j = n \hbox{ and none of }  c_n, \dots, c_{m+1} \hbox{ is } 1, \quad 
(\mbox{taking care of (b.i)})\cr
0 & \hbox{ if } 1 \le j \le m \hbox{ and } c_j = \tilde\ell_j = 1, 
\hskip .65in (\mbox{taking care of (b.ii)})\cr 
c_j &  \hbox{ otherwise}. \end{cases}
\end{equation}
Because at least one of 
$\tilde c_n, \dots, \tilde c_{m+1}$ is 1, for $1 \le j \le 2^m$
the $(j,j)$ entries will not interact with other
$(k,j)$ entry with $1 \le k \le N/2$ and $k\ne j$. 

Note also that a $(c_n\cdots c_1)$-gate with
$c_n, \dots, c_1 \in \{*,0,1,V\}$
corresponding to the unitary matrix
$$\tilde V = I_N + V_{n} \otimes\cdots \otimes V_1,$$
where 
$$V_i = \begin{cases} |0\ra \la 0| & \hbox{ if } c_i = 0,\cr
|1\ra \la 1| & \hbox{ if } c_i = 1,\cr
V - I_2 & \hbox{ if } c_i = V,\cr
I_2  & \hbox{ if } c_i = *.\cr
\end{cases}
$$
For the $(c_n\cdots c_1)$-gates used in the first columns, we have
$c_n, \dots, c_1 \in \{*,1,V\}$ with $c_1 \ne 1$.
So, changing the $1$-control in the $c_i$ position whenever $\tilde{\ell}_i=1$ 
in our rule is equivalent to applying a unitary similarity transform to change 
$\tilde V$ to $P_\ell^{t} V P_\ell$, where
$P_\ell$ is the permutation matrix changing the basis
$\{|j_n \cdots j_1\ra: j_r \in \{0,1\}\}$ to
$\{|j_n \dots j_1 \oplus \tilde \ell_n \dots \tilde \ell_1\ra: 
j_r \in \{0, 1\}\}$,
where $\tilde \ell_n \cdots \tilde \ell_1$ is the binary number corresponding to $\ell$.

So, the modified gates can be used  to annihilate 
$(d_j\oplus \ell,\ell)$ entries for $j = 1, \dots, N/2-1$.
After that,
only the $(\ell,\ell)$ and 
$(N/2+\ell,\ell)$ 
entries are nonzero in column 
$\ell$. We annihilate the $({N}/{2}+\ell,\ell)$ entry using the
$(V\hat c_{n-1}\dots\hat c_1)$-gate
to ensure that the annihilation in these steps will not affect the zero 
entries in the previous steps,
where $(\hat c_{n}-1 \cdots \hat c_1)$ is obtained from the 
binary sequence correspondence $(\tilde{\ell}_{n-1}\dots \tilde{\ell}_1)$ 
of $\ell$ by changing all 0 terms to $*$.\footnote{
For example, for Column 2 we change $(c_n \cdots c_1)$  to
$G_2(c_n \cdots c_1)$ by changing only $c_1$ and $c_n$ because $2$ 
corresponds to $0\cdots 01$, and $(\hat c_n \cdots \hat c_1) 
= (V * \cdots * 1)$; for Column 3, we change $(c_n \cdots c_1)$ to 
$G_3(c_n \cdots c_1)$ by
changing only $c_2$ and $c_n$ because $3$ corresponds to $0\cdots 010$, 
and $(\hat c_n \cdots \hat c_1) = (V * \cdots * 1 *)$; for Column 4, we change 
$(c_n \cdots c_1)$ to $G_4(c_n \cdots c_1)$ by
changing only $c_1,c_2$ and $c_n$ because $4$ 
corresponds to $0\cdots 011$, and $(\hat c_n \cdots \hat c_1) 
=(V * \cdots * 1 1)$.}

Note also that except for the last step one will always get the same set of 
$(c_n\cdots c_1)$-gates for the the elimination
of the lower half of the entries in  
Columns $2k-1$ and $2k$ because the modification in (\ref{tcj}) will have the 
same effects in these columns. This follows from the fact that the 
$(c_n \cdots c_1)$-gates for Column 1 satisfy 
$c_n, \dots, c_1 \in \{*,1,V\}$ with $c_1 \ne 1$.

\medskip
The annihilation steps of Column $\ell$ can be summarized in the following.

\medskip
\centerline{\bf Procedure 2.2}.

$\begin{array}{|l|}
\hline
\\
\hbox{Suppose in the } (n-1)-\hbox{qubit case, the off-diagonal
entries in Column $\ell$ are eliminated in the order of }\\ \\
\hskip 1in 
(a_1,\ell), \dots, (a_{N/2-\ell},\ell) \quad \hbox{ by } \quad 
{\bf D_1}-\hbox{gate}, \dots  
{\bf D_{N/2-\ell}}-\hbox{gate}. \\ \\
\hbox{For the } n-\hbox{qubit case,  eliminate the entries in the 
upper half of the column } \hbox{ in the order of }
\\ \\
\hskip .7in 
(a_1,\ell), \dots, (a_{N/2-1},\ell) \quad \hbox{ by } 
\quad (*{\bf D_1})-\hbox{gate}, \dots  
(*{\bf D_{N/2-\ell}})-\hbox{gate}.
\\ \\
\hbox{For } \mathbf{C}=(c_{n-1} \cdots c_1) \hbox{ let }
{G_\ell}(\mathbf{C}) =  (\tilde c_n \cdots \tilde c_1)
\hbox{  satisfy  (\ref{tcj})}, 
\hbox{ and let } d_i \hbox{ and } {\cal G}({\bf C}_i) 
\hbox{ be defined as in Procedure 1.1}.  
\\ \\
\hbox{Eliminate the entries in the lower half of the column
in the order of }
\\ \\
\hskip .7in
(d_1\oplus \ell, \ell), \dots, (d_{N/2-1+\ell},\ell) 
\quad \hbox{ by } \quad {G_\ell({\cal G}{(\bf C_1}))-\hbox{gate}, 
\dots   G_\ell({\cal G}({\bf C_{N/2-1}}))-\hbox{gate}}; \\ \\
\hbox{eliminate the } (N/2+\ell,\ell) \hbox{ entry by a }
(V\tilde c_{n-1}\cdots \tilde c_1)-\hbox{gate, where} 
(\hat c_{n-1} \cdots \hat c_1) \hbox{ is obtained from the binary} \\ \\

\hbox{sequence correspondence }
(\tilde{\ell}_{n-1}\dots \tilde{\ell}_1) \hbox{ of } \ell \hbox{ by changing 
 all 0 terms to } *.  \\ \\ 
\hline
\end{array}
$

\newpage
\bigskip\noindent
Several remarks concerning Procedures 1 and 2 are in order.
\begin{enumerate}
\item
In Column 1, it is easy to determine the order 
of the entries to be eliminated and the $(c_n \cdots c_1)$-gates used.

\item  For the lower half of Column $\ell$ with $2 \le \ell \le N/2$, we change the
order of entries to be eliminated to
$(d_1 \oplus \ell,\ell),  \dots, (d_{N/2}\oplus \ell, \ell)$, and 
change the $(c_n\cdots c_1)$-gates to $G_\ell(c_n\cdots c_1)$-gates.

\item The $(c_n \cdots c_1)$-gates used in Column 1 satisfy 
$c_n, \dots, c_1 \in \{*,1,V\}$ with  $c_1 \ne 1$.

\item The $(c_n\cdots c_1)$-gates used to eliminate the entries
in the lower half of Column $2k-1$ and $2k$ 
are always the same before the last step, for $k = 1,\dots, N/4$.

\item The $(c_n\cdots c_1)$-gates used in the last steps of 
Columns $1, \dots, N/2$ 
satisfy $c_n = V$, and $(c_{n-1}\cdots c_1)$ is obtained from the
binary sequences 
$(0\cdots 0), \dots, (1\cdots 1)$ 
of length $n-1$ by replacing  $0$ with $*$.
\end{enumerate}

The recurrence scheme easy to do. Even the most non-trivial steps of
adapting the procedures of eliminating the entries in the lower half
of the first column to other columns are quite straight forward.
We illustrate this for the case $n=4$.

\medskip\noindent
\large\textbf{Four qubit case, lower left block}\normalsize\vspace{0.2cm}

\medskip
\noindent
\textbf{Col 1, steps 8-15}\qquad
$\begin{array}{|c|c|c|c|c|c|c|c|c|}
\hline
\hbox{ entries } & \mbox{ \footnotesize\tt (10,1)} & \mbox{ \footnotesize\tt (12,1)} & \mbox{ \footnotesize\tt (11,1)} & \mbox{ \footnotesize\tt (14,1)} & \mbox{ \footnotesize\tt (16,1)} & \mbox{ \footnotesize\tt (15,1)} & \mbox{ \footnotesize\tt (13,1)} & \mbox{ \footnotesize\tt (9,1)} \cr\hline 
\hbox{ binary } & \mbox{ \footnotesize\tt 1001} & \mbox{\footnotesize \tt 1011} & \mbox{\footnotesize \tt 1010} & \mbox{\footnotesize \tt 1101} & \mbox{\footnotesize\tt 1111} & \mbox{\footnotesize\tt 1110} & \mbox{\footnotesize\tt 1100} & \mbox{\footnotesize\tt 1000} \cr\hline 
\hbox{ gates } & \mbox{ \footnotesize\tt 1**V} & \mbox{ \footnotesize\tt **1V} & \mbox{ \footnotesize\tt 1*V*} & \mbox{ \footnotesize\tt *1*V} & \mbox{ \footnotesize\tt **1V} & \mbox{ \footnotesize\tt *1V*} & \mbox{ \footnotesize\tt 1V**} & \mbox{ \footnotesize\tt V***} \cr\hline
\end{array}$\vspace{0.3cm}\\
\textbf{Col 2, steps 7-14}\qquad
$\begin{array}{|c|c|c|c|c|c|c|c|c|}
\hline
\hbox{ entries } & \mbox{ \footnotesize\tt (9,2)} & \mbox{ \footnotesize\tt (11,2)} & \mbox{ \footnotesize\tt (12,2)} & \mbox{ \footnotesize\tt (13,2)} & \mbox{ \footnotesize\tt (15,2)} & \mbox{ \footnotesize\tt (16,2)} & \mbox{ \footnotesize\tt (14,2)} & \mbox{ \footnotesize\tt (10,2)} \cr\hline 
\hbox{ binary } & \mbox{ \footnotesize\tt 1000} & \mbox{ \footnotesize\tt 1010} & \mbox{ \footnotesize\tt 1011} & \mbox{ \footnotesize\tt 1100} & \mbox{ \footnotesize\tt 1110} & \mbox{ \footnotesize\tt 1111} & \mbox{ \footnotesize\tt 1101} & \mbox{ \footnotesize\tt 1001} \cr\hline 
\hbox{ gates } & \mbox{ \footnotesize\tt 1**V} & \mbox{ \footnotesize\tt **1V} & \mbox{ \footnotesize\tt 1*V*} & \mbox{ \footnotesize\tt *1*V} & \mbox{ \footnotesize\tt **1V} & \mbox{ \footnotesize\tt *1V*} & \mbox{ \footnotesize\tt 1V**} & \mbox{ \footnotesize\tt V**1} \cr\hline
\end{array}$\vspace{0.3cm}\\
\textbf{Col 3, steps 6-13}\qquad
$\begin{array}{|c|c|c|c|c|c|c|c|c|}
\hline
\hbox{ entries } & \mbox{ \footnotesize\tt (12,3)} & \mbox{ \footnotesize\tt (10,3)} & \mbox{ \footnotesize\tt (9,3)} & \mbox{ \footnotesize\tt (16,3)} & \mbox{ \footnotesize\tt (14,3)} & \mbox{ \footnotesize\tt (13,3)} & \mbox{ \footnotesize\tt (15,3)} & \mbox{ \footnotesize\tt (11,3)} \cr\hline
\hbox{ binary } & \mbox{ \footnotesize\tt 1011} & \mbox{ \footnotesize\tt 1001} & \mbox{ \footnotesize\tt 1000} & \mbox{ \footnotesize\tt 1111} & \mbox{ \footnotesize\tt 1101} & \mbox{ \footnotesize\tt 1100} & \mbox{ \footnotesize\tt 1110} & \mbox{ \footnotesize\tt 1010} \cr\hline  
\hbox{ gates } & \mbox{ \footnotesize\tt 1**V} & \mbox{ \footnotesize\tt 1*0V} & \mbox{ \footnotesize\tt 1*V*} & \mbox{ \footnotesize\tt *1*V} & \mbox{ \footnotesize\tt 1*0V} & \mbox{ \footnotesize\tt *1V*} & \mbox{ \footnotesize\tt 1V**} & \mbox{ \footnotesize\tt V*1*} \cr\hline
\end{array}$\vspace{0.3cm}\\
\textbf{Col 4, steps 5-12}\qquad
$\begin{array}{|c|c|c|c|c|c|c|c|c|}
\hline
\hbox{ entries } & \mbox{ \footnotesize\tt (11,4)} & \mbox{ \footnotesize\tt (9,4)} & \mbox{ \footnotesize\tt (10,4)} & \mbox{ \footnotesize\tt (15,4)} & \mbox{ \footnotesize\tt (13,4)} & \mbox{ \footnotesize\tt (14,4)} & \mbox{ \footnotesize\tt (16,4)} & \mbox{ \footnotesize\tt (12,4)} \cr\hline 
\hbox{ binary } & \mbox{ \footnotesize\tt 1010} & \mbox{ \footnotesize\tt 1000} & \mbox{ \footnotesize\tt 1001} & \mbox{ \footnotesize\tt 1110} & \mbox{ \footnotesize\tt 1100} & \mbox{ \footnotesize\tt 1101} & \mbox{ \footnotesize\tt 1111} & \mbox{ \footnotesize\tt 1011} \cr\hline 
\hbox{ gates } & \mbox{ \footnotesize\tt 1**V} & \mbox{ \footnotesize\tt 1*0V} & \mbox{ \footnotesize\tt 1*V*} & \mbox{ \footnotesize\tt *1*V} & \mbox{ \footnotesize\tt 1*0V} & \mbox{ \footnotesize\tt *1V*} & \mbox{ \footnotesize\tt 1V**} & \mbox{ \footnotesize\tt V*11} \cr\hline
\end{array}$\vspace{0.3cm}\\
\textbf{Col 5, steps 4-11}\qquad
$\begin{array}{|c|c|c|c|c|c|c|c|c|}
\hline
\hbox{ entries } & \mbox{ \footnotesize\tt (14,5)} & \mbox{ \footnotesize\tt (16,5)} & \mbox{ \footnotesize\tt (15,5)} & \mbox{ \footnotesize\tt (10,5)} & \mbox{ \footnotesize\tt (12,5)} & \mbox{ \footnotesize\tt (11,5)} & \mbox{ \footnotesize\tt (9,5)} & \mbox{ \footnotesize\tt (13,5)} \cr\hline 
\hbox{ binary } & \mbox{ \footnotesize\tt 1101} & \mbox{ \footnotesize\tt 1111} & \mbox{ \footnotesize\tt 1110} & \mbox{ \footnotesize\tt 1001} & \mbox{ \footnotesize\tt 1011} & \mbox{ \footnotesize\tt 1010} & \mbox{ \footnotesize\tt 1000} & \mbox{ \footnotesize\tt 1100} \cr\hline 
\hbox{ gates } & \mbox{ \footnotesize\tt 1**V} & \mbox{ \footnotesize\tt 1*1V} & \mbox{ \footnotesize\tt 1*V*} & \mbox{ \footnotesize\tt 10*V} & \mbox{ \footnotesize\tt 1*1V} & \mbox{ \footnotesize\tt 10V*} & \mbox{ \footnotesize\tt 1V**} & \mbox{ \footnotesize\tt V1**} \cr\hline
\end{array}$\vspace{0.3cm}\\
\textbf{Col 6, steps 3-10}\qquad
$\begin{array}{|c|c|c|c|c|c|c|c|c|}
\hline
\hbox{ entries } & \mbox{ \footnotesize\tt (13,6)} & \mbox{ \footnotesize\tt (15,6)} & \mbox{ \footnotesize\tt (16,6)} & \mbox{ \footnotesize\tt (9,6)} & \mbox{ \footnotesize\tt (11,6)} & \mbox{ \footnotesize\tt (12,6)} & \mbox{ \footnotesize\tt (10,6)} & \mbox{ \footnotesize\tt (14,6)} \cr\hline 
\hbox{ binary } & \mbox{ \footnotesize\tt 1100} & \mbox{ \footnotesize\tt 1110} & \mbox{ \footnotesize\tt 1111} & \mbox{ \footnotesize\tt 1000} & \mbox{ \footnotesize\tt 1010} & \mbox{ \footnotesize\tt 1011} & \mbox{ \footnotesize\tt 1001} & \mbox{ \footnotesize\tt 1101} \cr\hline 
\hbox{ gates } & \mbox{ \footnotesize\tt 1**V} & \mbox{ \footnotesize\tt 1*1V} & \mbox{ \footnotesize\tt 1*V*} & \mbox{ \footnotesize\tt 10*V} & \mbox{ \footnotesize\tt 1*1V} & \mbox{ \footnotesize\tt 10V*} & \mbox{ \footnotesize\tt 1V**} & \mbox{ \footnotesize\tt V1*1} \cr\hline
\end{array}$\vspace{0.3cm}\\
\textbf{Col 7, steps 2-9\phantom{1}}\qquad
$\begin{array}{|c|c|c|c|c|c|c|c|c|}
\hline
\hbox{ entries } & \mbox{ \footnotesize\tt (16,7)} & \mbox{ \footnotesize\tt (14,7)} & \mbox{ \footnotesize\tt (13,7)} & \mbox{ \footnotesize\tt (12,7)} & \mbox{ \footnotesize\tt (10,7)} & \mbox{ \footnotesize\tt \phantom{1}(9,7)} & \mbox{ \footnotesize\tt (11,7)} & \mbox{ \footnotesize\tt (15,7)} \cr\hline
\hbox{ binary } & \mbox{ \footnotesize\tt 1111} & \mbox{ \footnotesize\tt 1101} & \mbox{ \footnotesize\tt 1100} & \mbox{ \footnotesize\tt 1011} & \mbox{ \footnotesize\tt 1001} & \mbox{ \footnotesize\tt 1000} & \mbox{ \footnotesize\tt 1010} & \mbox{ \footnotesize\tt 1110} \cr\hline   
\hbox{ gates } & \mbox{ \footnotesize\tt 1**V} & \mbox{ \footnotesize\tt 1*0V} & \mbox{ \footnotesize\tt 1*V*} & \mbox{ \footnotesize\tt 10*V} & \mbox{ \footnotesize\tt 1*0V} & \mbox{ \footnotesize\tt 10V*} & \mbox{ \footnotesize\tt 1V**} & \mbox{ \footnotesize\tt V11*} \cr\hline
\end{array}$\vspace{0.3cm}\\
\textbf{Col 8, steps 1-8\phantom{1}}\qquad
$\begin{array}{|c|c|c|c|c|c|c|c|c|}
\hline
\hbox{ entries } & \mbox{ \footnotesize\tt (15,8)} & \mbox{ \footnotesize\tt (13,8)} & \mbox{ \footnotesize\tt (14,8)} & \mbox{ \footnotesize\tt (11,8)} & \mbox{ \footnotesize\tt (9,8)} & \mbox{ \footnotesize\tt (10,8)} & \mbox{ \footnotesize\tt (12,8)} & \mbox{ \footnotesize\tt (16,8)} \cr\hline 
\hbox{ binary } & \mbox{ \footnotesize\tt 1110} & \mbox{ \footnotesize\tt 1100} & \mbox{ \footnotesize\tt 1101} & \mbox{ \footnotesize\tt 1010} & \mbox{ \footnotesize\tt 1000} & \mbox{ \footnotesize\tt 1001} & \mbox{ \footnotesize\tt 1011} & \mbox{ \footnotesize\tt 1111} \cr\hline
\hbox{ gates } & \mbox{ \footnotesize\tt 1**V} & \mbox{ \footnotesize\tt 1*0V} & \mbox{ \footnotesize\tt 1*V*} & \mbox{ \footnotesize\tt 10*V} & \mbox{ \footnotesize\tt 1*0V} & \mbox{ \footnotesize\tt 10V*} & \mbox{ \footnotesize\tt 1V**} & \mbox{ \footnotesize\tt V111} \cr\hline
\end{array}$

\newpage
\section{Total Number of Controls and Comparison to a Previous Study}

The following theorem gives the formula for the number  $g_n^k$ 
of $k$-controlled qubit gates used in the recurrence scheme of 
our decomposition for  a unitary matrix $U\in M_{2^n}$,
where $k = 0, 1, \dots, n-1$.

\begin{theorem}\label{counting} 
\begin{enumerate}
\item $g^0_n=n$
\item $g^{n-1}_n=\left\{\begin{array}{ll}
1 & \mbox{ if } n=1\\
4 & \mbox{ if } n=2\\
7+(n-3) & \mbox{ if } n\geq 3\\
\end{array} \right.$
\item $g^k_n=g^k_{n-1}+g^{k-1}_{n-1}+\binom{n-1}{k}$ for all $3\leq k <n-1$
\item $g^1_n=n(n-1)(2^{n-2}+1)$ for all $n\geq 2$

\item $g^2_n=\dfrac{1}{3}(4^n-4)-2^n(n-1)+\dfrac{n(n-1)(n-2)}{2}$ for all $n\geq 3$
\end{enumerate}
\end{theorem}

Note that $\sum\limits_{k=0}^{n-1}g_n^k=2^{n-1}(2^n-1)=N(N-1)/2$. By convention $g_1^0=1$. In general, if $n>1$, \[g_n^k=A_n^k+B_n^k+C_n^k+D_n^k,\] where $A_n^k$ is the number 
of $g_n^k$ gates used to annihilate entries in the upper left block of the matrix, $B_n^k$ is the 
number of $g_n^k$ gates used to annihilate entries of the lower half of columns $1,\ldots, 2^{n-1}$ 
excluding the entries of the form $(N/2+\ell,\ell)$. The number $C_n^k$ is the number of $g_n^k$ 
gates used to annihilate entries $(N/2+\ell,\ell)$, where $\ell\in\{1,\ldots, 2^{n-1}\}$. Finally 
$D_n^k$ is the number of $g_n^k$ gates used to annihilate the lower right block entries of the 
matrix. For example, we saw in section 2 that  
\[g_2^0=2=1+0+1+0 \qquad \mbox{ and } \qquad  g_2^1=4=0+2+1+1\] 
and 
\[g_3^0=3=2+0+1+0, \qquad g_3^1=18=4+10+2+2, \quad \mbox{ and } \quad g_3^2=0+2+1+4\]

\begin{remarks}\label{count} Immediately, we can see the following recursive properties.
\begin{enumerate}
\item $A_n^k=g_{n-1}^k$ for $k\in\{0,\ldots, n-2\}$ and $A_n^{n-1}=0$ as illustrated in tables 2 and 3. 
\item $D_n^k=g_{n-1}^{k-1}$ for $k\in\{1,\ldots, n-1\}$ and $D_n^{0}=0$ as seen from table 2 and 5. 
\item $C_n^k=\binom{n-1}{k}$ for $k\in\{0,\ldots,n-1\}$, because $C_k^n$ is the number of column indices $\ell$, with $1\leq \ell \leq 2^{n-1}$, such that the binary sequence of $\ell$ of length $n$ has exactly $k$ digits equal to $1$.
\item Observe that the gate 
$\mathbf{\cG_i}=\cG(\mathbf{C}_i)$, $1\leq i\leq \frac{N}{2}-1$, in table 1 
has exactly one $1$-control. All other gates accounted for by $B_n^k$ are obtained 
from the $\mathbf{\cG}_i$'s via the transformation 
$G_{\ell}$, for $2\leq \ell\leq \frac{N}{2}$. But notice that $G_{\ell}
(\mathbf{\cG}_i)$ either has the same number of controls as $\mathbf{\cG}_i$ or 
has one more control than $\mathbf{\cG}_i$. Hence $B_n^k=0$ for $k>2$ and 
$B_n^1+B_n^2=2^{n-1}(2^{n-1}-1)$.
\end{enumerate}
\end{remarks}

Let us observe the recursive scheme for the first column (see Table 1). The following lemma can be proven inductively from this scheme.

\begin{lemma}\label{lemcol1} If 
\[i=2^{s_1-1}+\sum\limits_{m=1}^{j}(2^{s_m-1}-1), \mbox{ where } 
1\leq s_1<s_2<\cdots <s_j\leq n-1 \mbox{ and } 1\leq j\leq n-1\] then 
\begin{equation}\label{ann1}
b_i=1+\sum\limits_{m=1}^{j}2^{s_m-1}, \quad \mbox{ and } \quad \mathbf{C}_i
=(*\cdots * c_{s_2}*\cdots * c_{s_1}* \cdots *),\end{equation} 
where $(c_{s_2},c_{s_1})=(*,V)$ when $j=1$, otherwise $(c_{s_2},c_{s_1})=(1,V)$.
\end{lemma}

\begin{lemma}\label{lem1}
Let $\mathbf{\cG}_1,\ldots, \mathbf{\cG}_{N/2-1}$ be as in remark 4.2.4. 
Suppose $\mathbf{\cG}_i$ is a $(c_{in}\ldots c_{i1})$-gate. Then the following 
holds
\[\#\{i| c_{ik}=1\}=\left\{\begin{array}{ll}
n-1 & \mbox{ when } k=n,\\
2^{n-k-1}(k-1) & \mbox{ otherwise}.
\end{array} \right. \]
\end{lemma}

\begin{proof} We want to know how many of the $\mathbf{\cG}_i$'s have a $1$-control 
in the $k^{th}$ bit. By Lemma \ref{lemcol1}, we know that the 
$\mathbf{\cG}_i$'s satisfying this annihilate entries $b_i$ of the form given in 
equation (\ref{ann1}), where $s_2=k$ and $s_j=n$. If $k=n$, then $j=2$ and thus we have 
$(n-1)$ choices for $s_1$. If $k<n$, we have $k-1$ choices for $s_1$ and we are free to 
choose which ones in $\{2^{k+1},\ldots, 2^{n-1}\}$ to include in the sum defining 
$b_i$. The conclusion then follows 
\end{proof}

Next, let us look at the gates used to annihilate entries of column 
$\ell\in\{1,\ldots, \frac{N}{2}\}$ that contribute to $B_n^1$.
 
\begin{lemma} \label{lem2} Let $2^{m-1}< \ell\leq 2^m$ with $1\leq m \leq n-1$ 
and $\mathbf{\cG}_1,\ldots, \mathbf{\cG}_{\frac{N}{2}-1}$ be as in 
Lemma \ref{lem1}. Then 
$$\#\{i| G_{\ell}(\mathbf{\cG}_i) \mbox{ has  exactly one control }\}  
=\left\{\begin{array}{ll} n-1 & \mbox{ if } m=n-1,\\
(n-1)+\sum\limits_{k=m+1}^{n-1} 2^{n-k-1}(k-1) & \mbox{ otherwise}. 
\end{array} \right.$$
\end{lemma}

\begin{proof} 
If $2^{m-1}<\ell\leq 2^m$, then $\ell = \sum_{j=1}^{m}\tilde \ell_j 2^{j-1} + 1$. Recall that $G_{\ell}(\mathbf{\cG}_i)$ has exactly one control if $\mathbf{\cG}_i=(c_{in},\ldots,c_{i1})$ has its one $1$-control in $\{c_{i(m+1)},\ldots, c_{in}\}$. Thus 
\[\#\{i| G_{\ell}(\mathbf{\cG}_i) \mbox{ has  exactly one control }\}=\bigcup_{k=m+1}^{n}\#\{i| c_{ik}=1\} \]
The conclusion follows from Lemma \ref{lem1}. 
\end{proof}

\noindent
{\bf Proof of Theorem \ref{counting}}
\begin{enumerate}
\item A control-free gate can only be utilized in Column 1. This is because 
when we transform the matrix to the form $[1]\oplus U'$, 
the succeeding gates must make  sure that the first row does not interact 
with other rows.  As mentioned in Lemma \ref{lemcol1} and 
illustrated in Table 1, these gates with no control are the gates 
that annihilate the 
entries of the form $(1+2^{s_m},1)$ for $m\in \{1,\dots n\}$. Indeed $g^0_n=n$.
\item We have shown that $g^0_1=1$, $g^1_2=4$ and $g^2_3=7$. 
From Remark \ref{count} we 
deduce that $g_n^{n-1}=\binom{n-1}{n-1}+g_{n-1}^{n-2}$ for all $n\geq 4$ and  
hence 
\[g^{n-1}_{n}=1+g^{n-2}_{n-1}=(n-3)+g^{2}_{3}=(n-3)+7.\]
\item Now, assume $n-1>k\geq 3$. From Remark \ref{count}, we get 
$g^k_n=g_{n-1}^k+\binom{n-1}{k}+0 + g_{n-1}^{k-1}$.
\item When $n=2$, we know that that $g^1_2=4=2(2-1)(2^0+1)$. 

Now, assume $n>2$. From Remark \ref{count}, $g^1_n=g_{n-1}^1+B_n^1+(n-1)+g_{n-1}^0$. 
Let us look at the summation defining $B_n^1$. From Remark \ref{count}.4, Column 1 
contributes $\frac{N}{2}-1=2^{n-1}-1$ gates to $B_n^1$.  From 
Lemma \ref{lem2}, we deduce  that 
\begin{equation}\label{B1}
\begin{array}{lcl}
B_n^1 & = &\left(2^{n-1}-1\right)+ 2^{n-2}(n-1) + 
\sum\limits_{m=1}^{n-2}2^{m-1}\left[(n-1)+ \sum\limits_{k=m+1}^{n-1}2^{n-k-1}(k-1) 
\right]\\
& = & \left(2^{n-1}-1\right)n+\left[2^{n-3}n(n-3)-2^{n-2}+n\right] = 2^{n-3}(n+2)(n-1).
\end{array}
\end{equation}
Thus $g^1_n-g^1_{n-1}=2(n-1)+2^{n-3}(n+3)(n+2)(n-1)$. Using a telescoping sum, we get 
\[\begin{array}{lcl}
g^1_n & = & g^1_2+\sum\limits_{m=3}^{n}\left[2(m-1)+2^{m-3}(m+3)(m+2)(n-1)\right]\\
& = & (2^{n-2}+1)(n)(n-1).
\end{array}
\]
\item If $n=3$, $g^2_{3}=7=\frac{1}{3}(4^3-4)-2^3(3-1)+\frac{3\cdot 2 \cdot 1}{2}$. 
Now, assume $n>3$. From Remark \ref{count} and equation (\ref{B1}),
\[g^2_{n}=g^2_{n-1}+g^1_{n-1}+\binom{n-1}{2}+2^{n-1}(2^{n-1}-1)-2^{n-3}(n+2)(n-1).\]
Then 
\[\begin{array}{lcl}
g^2_{n}-g^2_{n-1} & = & (2^{n-3}+1)(n-1)(n-2) +\dfrac{(n-2)(n-1)}{2} +
2^{n-1}(2^{n-1}-1)-2^{n-3}(n+2)(n-1)\\
& = & 2^{n-1}(2^{n-1}-n)+\frac{3}{2}(n-2)(n-1).
\end{array}
\]
And hence
\[\begin{array}{lcl}
g^2_{n}& = & g^2_{3} + \sum\limits_{m=4}^{n}\left[2^{m-1}(2^{m-1}-m)
+\frac{3}{2}(m-2)(m-1)\right]\\
& = & \dfrac{1}{3}(4^n-4)-2^n(n-1)+\dfrac{n(n-1)(n-2)}{2}.
\end{array}\]
\end{enumerate}
\vskip -.3in {\hfill $\Box$}

\vskip .5in

In \cite{fin}, the Gray code basis was utilized to achieve the same goal of this paper. 
Let us denote the total number of gates with $k$ controls in the decomposition scheme 
presented in \cite{fin} by $\mathbf{g}_n^{k}$. The recursion formula presented in the 
said study is 
\[\mathbf{g}_n^{k}=\mathbf{g}_{n-1}^{k}+\mathbf{g}_{n-1}^{k-1}+\mbox{max}(2^{n-2},2^k)+(2^{2n-k-2}-2^{n-2})\quad \mbox{ (for  $k\geq 1$)}\]
with the conditions that $\mathbf{g}_n^0=2^{n-1}$ and $\mathbf{g}_n^n=0$ for all $n$. Let us compare values for small $n$.

{\small
$$\begin{array}{|c|c|c|c|c|c|c|}
\hline n & g^0_n\ /\ \mathbf{g}^0_n & g^1_n\ /\ \mathbf{g}^1_n & g^2_n\ /\ \mathbf{g}^2_n & g^3_n\ /\ \mathbf{g}^3_n & g^4_n\ /\ \mathbf{g}^4_n & T_1(n)\ /\ T_2(n)\\\hline 
1 & 1 \ /\ 1  & - & - & - & - & 0\ /\ 0 \\\hline
2 & 2 \ /\ 2  & 4 \ /\ 4 & - & - & - & 4\ /\ 4 \\\hline
3 & 3 \ /\ 4  & 18 \ /\ 14 & 7 \ /\ 10 & - & - & 32\ /\ 34 \\\hline
4 & 4\ /\ 8   & 60 \ /\ 50 & 48 \ /\ 40 & 8 \ /\ 22 & - & 180 \ /\ 196 \\ \hline
5 & 5 \ /\ 16 & 180 \ /\ 186 & 242 \ /\ 154 & 60 \ /\ 94 & 9 \ /\ 46 & 880 \ /\ 960\\ \hline
\end{array}$$
}

\noindent
\begin{minipage}{90mm}
Here, $T_1(n)$ (respectively, and $T_2(n)$) is the total number of controls in 
the decomposition of a unitary $U\in M_{2^n}$ using the scheme in this paper 
(respectively, the scheme in \cite{fin}). Starting from $n=3$, we get a small 
advantage in our decomposition and because both methods are recursive, the 
discrepancy becomes large as $n$ gets larger. For example, $T_2(10)-T_1(10)=30,720$.

\medskip
In Figure 1, we plot the difference between $T_2$ and $T_1$ for $n$ 
from 1 to 50. We use the log scale in the $y$-axis.

\vskip .3in \ \
~
\end{minipage} {\hfill
\begin{minipage}{70mm}
\begin{center}
\includegraphics[scale=0.35]{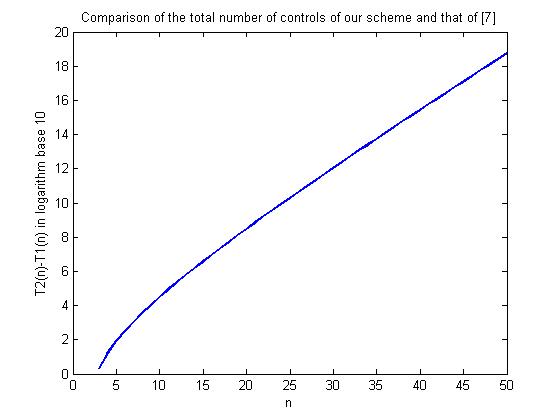}

{\bf Figure 1}
\end{center}
\end{minipage}
}

\section{Concluding Remarks and Future Research}

In this paper, we present a recurrence scheme for generating
controlled single qubit unitary gates 
$U_1, \dots, U_r$ with $r \le N(N-1)/2$ 
such that $U_r \cdots U_1 U = I_N$.
Consequently, $U = U_1^{\dag} \cdots U_r^{\dag}$. We have the following.

\medskip
\noindent{\bf Recurrence scheme}

\medskip
$\begin{array}{|l|}
\hline
\\
\hbox{
{\bf Step 1}
Partition } U \in M_n \hbox{ into a } 2\times 2 \hbox{
block matrix with each block lying 
in } M_{N/2}, \hbox{ where } N = 2^n. \\  
\\
\hbox{
{\bf Step 2} Use the scheme of the }
(n-1)-\hbox{ qubit case to help reduce } U \hbox{ to the form  }
I_{N/2} \oplus \tilde U \hbox{  with } \tilde U \in M_{N/2}.
\\ \\

\   \hbox{{\bf Step 2.1} For Column 1, use Prcedure 2.1 in Section 3.}
\\ \\
\  \hbox{{\bf Step 2.2} For Column } \ell \hbox{ with }
2 \le \ell \le N/2, \hbox{ use Procedure 2.2 in section 3.} \\ \\

\hbox{ {\bf Step 3} Apply the scheme of the }
(n-1)-\hbox{qubit case to transform } \tilde U \hbox{ to }I_{N/2}. \\ \\

\hline
\end{array}$

\medskip\noindent
It is worth noting that one can actually describe the entire
recursive scheme in terms of the steps used to eliminate the 
off-diagonal entries of the first column as follows.

\begin{itemize}
\item We first generate the $(c_n\cdots c_1)$-gates for eliminating
the off-diagonal entries:

For $n = 1$ use $V$ to eliminate the $(2,1)$ entry;
for $n > 1$ modify the $(c_{n-1}\cdots c_1)$-gates to
$(*c_{n-1}\cdots c_1)$-gates to eliminate the off-diagonal entries
in upper half of Column 1 in the $n$-qubit case,
and ${\cal G}(c_{n-1}\cdots c_1)$-gates to eliminate the 
entries in the lower half.

\item Once, we have the $(c_n\cdots c_1)$-gates for Column 1, 
we can modify them to eliminate the off-diagonal entries for
the leading $2^m\times 2^m$ blocks for $m = 1, \dots, n$,
using Steps 2.1 and 2.2 described in Section 3.
\end{itemize}

We give recursive formulas for the number of controlled
single qubit gates needed in the decomposition.
The total number of controls used in our scheme 
is less than that in \cite{fin}.

For future research, it might be interesting to design other recurrence
schemes, which are easy to implement and use even less controls.
Moreover, there might be other optimality criteria depending on the
physical implementation of qubits.
One may take this into consideration and assign a 
cost $w_k$ for implementing  $k$-controlled single qubit
gates, and then study the optimal decomposition
by minimizing the cost instead of number of controls.

Matlab programs for the decomposition using our scheme is posted 
at http://ckixx.people.wm.edu/mathlib.html.
The  program \emph{decomposition.m}  displays the order of entries to be annihilated, the $(c_nc_{n-1}\cdots c_1)$-gate used and the single qubit gate $V\in M_2$ used for the controlled gates. One types 
\texttt{
[U,A,x,y,controls,num,V]=decomposition(n);}
in the command line, where $n$ is the number of bits and the program will prompt for the user to 
either choose to create a random unitary matrix or input the unitary matrix manually. The variable 
$\mathtt U$ is the unitary matrix being decomposed. the variable $\mathtt A$ is an array of 
strings that describe the form of the gate, and $\mathtt{ (x,y)}$ are the row and column indices 
arranged according to their order of annihilation. The variable \texttt{controls} gives the total 
number of controls used and $\mathtt{num}$ is the number of nontrivial unitary controlled gates 
used. The variable $\mathtt V$ is the product of the controlled gates and hence, must always equal 
to the identity matrix. This is used to help the user check that the decomposition is correct. The 
matlab program \emph{gatecount.m} counts the total number of controls in our scheme and that of 
\cite{fin}. 

\bigskip\noindent
{\bf Acknowledgment}

\medskip
The research of Li
was supported by a USA NSF grant, and a HKU RGC grant.
Li was  an honorary professor of the
University of Hong Kong, and an honorary professor of Shanghai University.

After the first version of this manuscript was put on arXiv,  
the following related references \cite{Bet,Cet,dask,mdel,slepoy,Tucci} came to our attention.

\end{document}